\documentclass{article}
\usepackage{spconf,amsmath,amssymb,amsthm,graphicx,booktabs,url}
\usepackage[hidelinks]{hyperref}

\newtheorem{proposition}{Proposition}
\newtheorem{corollary}{Corollary}

\title{When EER Hides Deployment Failure: Auditing Threshold Transfer and Unlabeled Score Calibration for Speech Deepfake Detectors}

\name{Jingwen Zhou \qquad Mingzhe Wang}
\address{Xidian University, Xi'an, China}

\begin{document}
\ninept
\maketitle

\begin{abstract}
Speech deepfake countermeasures (CMs) are compared almost exclusively by equal error rate (EER), a metric computed at an oracle threshold chosen on the labeled test set. Deployed CMs enjoy no such oracle: a threshold must be fixed in advance and applied to unlabeled target data. We audit this gap with a frozen state-of-the-art SSL-AASIST detector trained on ASVspoof~2019~LA. While its in-domain EER is 0.21\%, transferring its LA-calibrated threshold to the In-the-Wild corpus yields a half total error rate (HTER) of 39.5\%, with 78.7\% of bona fide speech rejected---even though the In-the-Wild EER (11.2\%) appears moderate. We then test whether popular unlabeled test-time corrections close this gap, and first prove a simple proposition: any strictly increasing score transform, including z-norm, temperature/shift calibration, and embedding mean alignment under a frozen linear head, cannot change EER. An audit of seven corrections on In-the-Wild and ASVspoof~2021~DF confirms the proposition empirically and exposes two further failure modes: AS-norm with an unlabeled target cohort collapses (EER 11.2\%$\to$60.2\%), and pseudo-label calibration that reduces HTER by 38\% relative on In-the-Wild degenerates to 50\% HTER on DF21, whose spoof prior is 96\%. No audited correction reduces EER by more than 1\% relative. We recommend reporting HTER at a transferred threshold alongside EER.
\end{abstract}


\section{Introduction}
\label{sec:intro}

Speech deepfake countermeasures have evolved rapidly, from hand-crafted spectral features \cite{todisco2017cqcc} through end-to-end raw-waveform networks \cite{tak2021rawnet2,jung2022aasist} to self-supervised front-ends \cite{baevski2020wav2vec,babu2022} fine-tuned for spoofing detection \cite{wang2022odyssey,tak2022}, with progress tracked by the ASVspoof challenge series \cite{wu2015asvspoof,todisco2019,yamagishi2021,wang2024asvspoof5}. State-of-the-art systems now reach near-zero EER on the benchmark they are trained on, with self-supervised front-ends coupled to graph attention back-ends reporting below 0.3\% EER on ASVspoof~2019~LA \cite{tak2022,jung2022aasist}. Yet cross-corpus fragility has been reported consistently for a decade \cite{wu2015survey,korshunov2016,chen2020odyssey,zhang2021oneclass}: the same systems degrade by one to two orders of magnitude on out-of-domain data such as the In-the-Wild (ITW) corpus \cite{muller2022} or multilingual collections \cite{muller2024mlaad}, and train-time augmentation \cite{tak2022rawboost} cannot anticipate every deployment shift. This generalization gap is well documented. Far less appreciated is that \textbf{EER systematically understates the resulting deployment damage}, because EER is evaluated at an oracle threshold selected \emph{a posteriori} on the labeled test set. A deployed CM has no oracle: its threshold must be fixed in advance, typically on labeled source-domain data, and applied unchanged to unlabeled target traffic.

This paper audits that scenario end-to-end. We freeze the official SSL-AASIST model \cite{tak2022} trained on ASVspoof~2019~LA \cite{todisco2019,wang2020asvspoof}, fix its threshold at the LA EER operating point, and measure the half total error rate (HTER) when this threshold is transferred to ITW \cite{muller2022} and ASVspoof~2021~DF \cite{yamagishi2021}. The result is stark: an ITW EER of 11.2\%---a number suggesting a degraded but usable system---coexists with a transferred-threshold HTER of 39.5\%, driven by a 78.7\% false rejection rate on bona fide speech. The dominant failure is the operating point, not the ranking (Fig.~\ref{fig:dist}).

Practitioners facing this problem commonly turn to \emph{unlabeled} test-time corrections imported from neighboring fields: score normalization from speaker verification \cite{auckenthaler2000,cumani2011,matejka2017}, temperature scaling from probability calibration \cite{guo2017}, correlation alignment (CORAL) from the domain adaptation literature \cite{sun2016}, and batch-norm statistics adaptation from test-time adaptation \cite{schneider2020,wang2021tent}. Our central question is which of these can---in principle and in practice---repair the EER, the operating point, or both. Our contributions are:
\begin{itemize}\setlength\itemsep{1pt}
\item a formal monotone-invariance proposition showing that a large family of unlabeled score calibrations (z-norm, temperature/shift, embedding mean alignment under a frozen linear head) cannot change EER, only the operating point (Sec.~\ref{sec:prop});
\item an audit of seven unlabeled corrections on a frozen SSL-AASIST over ITW and ASVspoof~2021~DF, reporting EER and transferred-threshold HTER under one consistent protocol (Secs.~\ref{sec:methods}--\ref{sec:exp});
\item three empirically documented failure modes: provable EER invariance; collapse of AS-norm with an unlabeled target cohort (EER $\to$ 60\%), a negative transfer from ASV practice; and strong sensitivity of pseudo-label calibration to the unknown target prior (HTER 24.4\% on ITW vs.\ a degenerate 50\% on DF21);
\item a concrete recommendation: report HTER at a transferred threshold alongside EER.
\end{itemize}

\textbf{Related work.} Borodin \emph{et al.} \cite{borodin2026} evaluate public CMs at an EER threshold calibrated on pooled external benchmarks; their audit, complementary to ours, concerns spoof-only multilingual data and reports spoof rejection only, without bona fide traffic, calibration methods, or invariance analysis. Khan \emph{et al.} \cite{khan2026trace} report gaps between ``transfer EER'' and ``free EER'' for partial-deepfake detection and address them through feature design rather than test-time correction. Test-time adaptation of deepfake detectors has been studied for \emph{images} \cite{nguyenle2025}; we are not aware of a systematic audit of unlabeled test-time corrections for speech CMs.

\section{Threshold Transfer and Monotone Invariance}
\label{sec:prop}

\textbf{Setup.} A CM assigns a scalar score $s(x)\in\mathbb{R}$, higher meaning more bona fide; throughout, $s=\ell_{\text{bona}}-\ell_{\text{spoof}}$ is the logit difference. On a labeled test set, $\mathrm{FRR}(\tau)$ and $\mathrm{FAR}(\tau)$ denote the bona fide rejection and spoof acceptance rates at threshold $\tau$; EER is their common value at the threshold where they coincide. Threshold-aware alternatives exist---the tandem t-DCF \cite{kinnunen2018tdcf} and the calibration-sensitive $C_{\mathrm{llr}}$ \cite{brummer2006}---but EER remains the de facto headline metric in CM research. We adopt the \emph{threshold transfer} protocol: $\tau_{\mathrm{src}}$ is fixed at the EER point of the labeled \emph{source} evaluation set, and on the target we report $\mathrm{HTER}(\tau_{\mathrm{src}})=\tfrac12[\mathrm{FRR}(\tau_{\mathrm{src}})+\mathrm{FAR}(\tau_{\mathrm{src}})]$ together with its FRR/FAR decomposition. Concurrent work adopts the same protocol to probe cross-lingual generalization on spoof-only corpora \cite{lrlspoof2026}; we instead use it to isolate the operating-point component of the cross-domain failure and to audit which unlabeled corrections can repair it.

\begin{proposition}[Monotone invariance of EER]
\label{prop:mono}
Let $g:\mathbb{R}\to\mathbb{R}$ be strictly increasing on the range of $s$. Then $g\circ s$ and $s$ produce identical ROC curves on any test set, and in particular identical EER. This holds regardless of how $g$ is estimated, including from unlabeled target data.
\end{proposition}
\begin{proof}
For every $\tau$, $\{x: g(s(x))\ge\tau\}=\{x: s(x)\ge g^{\dagger}(\tau)\}$, where $g^{\dagger}$ is the generalized inverse of $g$. Both scores therefore realize exactly the same family of decision sets, hence the same set of achievable $(\mathrm{FRR},\mathrm{FAR})$ pairs.
\end{proof}

\begin{corollary}
\label{cor:family}
The following unlabeled corrections leave EER unchanged:
(a)~\emph{z-norm}, $s'=(s-\hat\mu_T)/\hat\sigma_T$ with unlabeled target statistics;
(b)~\emph{temperature/shift calibration}, $s'=(s+b)/T$ with $T>0$, however $T$ and $b$ are fitted;
(c)~\emph{embedding mean alignment}, $e'=e-\hat\mu_T+\hat\mu_S$, under a frozen linear head, since it adds a constant to each logit and is therefore a positive affine map of $s$.
\end{corollary}

The corollary covers much of what practitioners try first. It implies that an EER improvement reported for such a method can only stem from an implementation artifact or from label leakage, and conversely that the genuine value of these methods, if any, must appear in \emph{threshold-dependent} metrics. Corrections that can move EER must act non-monotonically on the score: per-utterance cohort normalization (AS-norm \cite{matejka2017}, whose statistics depend on the utterance), full-covariance CORAL \cite{sun2016} (which mixes embedding dimensions before the linear head), or batch-norm statistics adaptation \cite{schneider2020} (which renormalizes internal features channel-wise).

\section{Audited Corrections}
\label{sec:methods}

All corrections use no target labels and keep all model weights frozen. For threshold transfer, each correction is applied identically to the source evaluation scores (or the source set is re-scored by the adapted model), and $\tau$ is re-derived at the source EER point in the corrected score space.

\textbf{(C1) z-norm.} Global standardization of $s$ with unlabeled target mean and standard deviation.

\textbf{(C2) Temperature/shift.} The top and bottom quartiles of target scores are pseudo-labeled bona fide and spoof, respectively; $(T,b)$ are then fitted by logistic negative log-likelihood and decisions are taken at $s'=(s+b)/T=0$. Since the pseudo-classes are separable by construction, the likelihood has no finite minimizer; we fix the procedure to a single deterministic L-BFGS run (100 iterations from $T{=}1$, $b{=}0$), whose endpoint defines the calibration.

\textbf{(C3) Mean alignment.} Embedding-space mean shift to the source statistics, re-scored by the frozen linear output layer acting on the 160-dimensional AASIST readout.

\textbf{(C4) CORAL.} Whitening--recoloring \cite{sun2016} of the embeddings, $e'=(e-\hat\mu_T)\,\hat\Sigma_T^{-1/2}\hat\Sigma_S^{1/2}+\hat\mu_S$ with shrinkage $10^{-4}I$, re-scored by the frozen head.

\textbf{(C5) AS-norm, target cohort.} Per-utterance normalization $s'_i=(s_i-\mu_i)/\sigma_i$, where $\mu_i,\sigma_i$ are computed over the scores of the $k$ unlabeled target utterances most cosine-similar to utterance $i$ in embedding space (leave-one-out), following adaptive cohort selection in ASV \cite{matejka2017}.

\textbf{(C6) AS-norm, source bona fide cohort.} As C5, but the cohort is the labeled \emph{source} bona fide set; no target labels are used.

\textbf{(C7) BN statistics adaptation.} Running means and variances of the 20 BatchNorm layers in the AASIST back-end are re-estimated on unlabeled target data (200 batches of 32, cumulative averaging); all weights and the wav2vec~2.0 front-end are untouched. This is the normalization component of TENT \cite{wang2021tent,schneider2020}.

C1--C3 are monotone (Cor.~\ref{cor:family}); C4--C7 are not.

\section{Experiments}
\label{sec:exp}

\subsection{Setup}
\label{ssec:setup}

\textbf{Model.} SSL-AASIST \cite{tak2022}: a wav2vec~2.0 XLS-R 300M front-end \cite{babu2022} with an AASIST graph-attention back-end \cite{jung2022aasist}, using the authors' released checkpoint trained on ASVspoof~2019~LA (\texttt{LA\_model.pth}, obtained from a public mirror of the official release). Inputs are 64{,}600 samples (about 4\,s) at 16\,kHz, cropped or tiled as in the original recipe; no augmentation is applied at test time. Scores are bona-minus-spoof output logits; embeddings are the 160-dimensional penultimate readout. Inference uses a batch size of 32 in full precision.

\textbf{Data.} \emph{Source:} ASVspoof~2019~LA eval \cite{todisco2019,wang2020asvspoof}; we use a 31{,}662-utterance subset (3{,}300 bona fide / 28{,}362 spoof) corresponding to four of the nine shards of a public parquet mirror, with every label verified against the official protocol file (zero mismatches). Our in-domain EER on this subset is 0.21\%, consistent with the 0.22\% reported by the authors \cite{tak2022}, which validates the pipeline. \emph{Targets:} the complete In-the-Wild corpus \cite{muller2022} (31{,}779 utterances; 19{,}963 bona fide / 11{,}816 spoof; 62.8\% bona fide prior), and an ASVspoof~2021~DF eval subset \cite{yamagishi2021} of 30{,}592 utterances (1{,}149 bona fide / 29{,}443 spoof; 3.8\% bona fide prior) from four of the 80 shards of a public parquet mirror, chosen with a uniform stride and fixed before scoring.

\textbf{Metrics.} EER (oracle threshold) and HTER at the transferred threshold, with FRR/FAR decomposition. The transferred threshold for the uncorrected system is $\tau_{\mathrm{LA}}=4.97$ (LA EER point).

\textbf{Sensitivity to the discrete EER crossing.} With 3{,}300 bona fide trials, the empirical FRR is quantized in steps of $0.03\%$, so the LA EER condition $\mathrm{FRR}=\mathrm{FAR}$ is attained on a short interval of thresholds rather than at a unique point; $\tau_{\mathrm{LA}}=4.97$ is the observed-score grid point minimizing $|\mathrm{FRR}-\mathrm{FAR}|$. Selecting the opposite end of the crossing interval ($\tau=5.44$) moves the ITW baseline HTER from 39.49\% to 40.20\% (FRR $78.7\%\to80.2\%$) and the DF21 baseline HTER from 17.05\% to 17.17\%---under one HTER point in the worst case, leaving every conclusion unchanged. All transferred-threshold numbers below are reported at the stated $\tau_{\mathrm{LA}}$ so that each entry is exactly reproducible.

\begin{figure}[t]
\centering
\includegraphics[width=\linewidth]{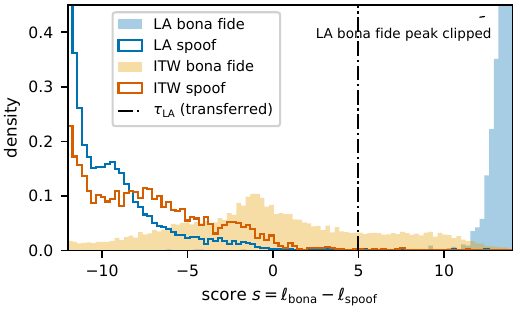}
\caption{Score distributions of the frozen CM on its source domain (ASVspoof~2019~LA eval) and on In-the-Wild (ITW), with the LA-calibrated threshold $\tau_{\mathrm{LA}}$ shown dash-dotted (the LA bona fide peak, density $1.2$, is clipped for readability). The ITW bona fide mass lies largely below $\tau_{\mathrm{LA}}$: 78.7\% of genuine speech is rejected although the ITW EER is a moderate 11.2\%.}
\label{fig:dist}
\end{figure}

\subsection{Results on In-the-Wild}
\label{ssec:itw}

\begin{table}[t]
\caption{Frozen SSL-AASIST on In-the-Wild (full corpus, 31{,}779 utterances). EER uses the oracle threshold; HTER uses the threshold transferred from ASVspoof~2019~LA after the identical correction (Sec.~\ref{sec:methods}). ``M'' marks corrections that are provably monotone (Cor.~\ref{cor:family}). All values are measured.}
\label{tab:itw}
\centering
\footnotesize
\setlength{\tabcolsep}{3.5pt}
\begin{tabular}{@{}llcccc@{}}
\toprule
& Method & EER\,\% & HTER\,\% & FRR\,\% & FAR\,\% \\
\midrule
& LA eval (in-domain, ref.) & 0.21 & -- & -- & -- \\
\midrule
& baseline (no correction) & 11.18 & 39.49 & 78.7 & 0.3 \\
\midrule
M & C1 z-norm & 11.18 & 45.02 & 89.9 & 0.1 \\
M & C2 temperature/shift & 11.18 & \textbf{24.39} & 47.3 & 1.5 \\
M & C3 mean alignment & 11.18 & 46.65 & 93.2 & 0.1 \\
\midrule
& C4 CORAL & \textbf{11.09} & 43.69 & 87.2 & 0.2 \\
& C5 AS-norm (target, $k{=}100$) & 60.16 & 60.70 & 77.5 & 43.9 \\
& C6 AS-norm (src bona, $k{=}100$) & 12.11 & 38.63 & 76.1 & 1.2 \\
& C7 BN adaptation & 11.23 & 39.26 & 78.3 & 0.3 \\
\bottomrule
\end{tabular}
\end{table}

Table~\ref{tab:itw} first confirms Prop.~\ref{prop:mono} \emph{exactly}: C1--C3 reproduce the baseline EER to the last digit, which also serves as an empirical check that the implementation is leakage-free. The non-monotone corrections are free to move EER, and essentially do not: CORAL improves it by under 1\% relative ($11.18\to11.09$), BN adaptation leaves it flat at $11.23$ (the in-domain LA EER is preserved at 0.21\% after adaptation, indicating that the domain gap does not reside in back-end BN statistics), and both AS-norm variants \emph{degrade} the ranking. With larger cohorts, C5 remains collapsed ($k{=}300$: 59.7\%, $k{=}1000$: 59.3\%) and C6 converges back toward the baseline (11.30\%, 11.21\%).

The operating-point picture is entirely different. The raw transferred threshold rejects 78.7\% of genuine speech. Pseudo-label temperature/shift calibration (C2)---which provably cannot change EER---reduces HTER from 39.5\% to 24.4\%, a 38\% relative reduction and the largest deployment gain in the audit. EER-centric reporting makes this method invisible, and makes the baseline's deployment failure invisible, at the same time.

\textbf{Failure mode I: provable invariance.} As stated by Prop.~\ref{prop:mono}, half of the audited toolbox cannot move EER by construction; this is exactly reproduced empirically.

\textbf{Failure mode II: cohort contamination.} AS-norm with an unlabeled target cohort (C5) collapses to 60.2\% EER. In ASV the normalization cohort is impostor-only by construction; in a CM, the unlabeled target cohort is a 63/37 bona/spoof mixture, so the per-utterance statistics are dominated by the test utterance's own class neighborhood, and the normalization removes class separation rather than nuisance offsets. The FRR/FAR decomposition makes the mechanism visible: spoof acceptance explodes (FAR $0.3\%\to43.9\%$) while FRR barely moves---exactly what is expected when each spoof trial is normalized against a neighborhood of other spoof trials. The source bona fide cohort (C6) is safe but inert for the opposite reason: its per-trial statistics vary little across test trials, so the transform approaches a global affine map---EER-invariant in the limit by Cor.~\ref{cor:family}---and the measured EER drifts back to the baseline as the cohort grows ($12.11\%\to11.30\%\to11.21\%$ for $k{=}100/300/1000$, against an $11.18\%$ baseline). ASV-style score normalization should not be transplanted to CMs without rethinking the cohort.

\subsection{Cross-dataset check on ASVspoof 2021 DF}
\label{ssec:df}

\begin{figure}[t]
\centering
\includegraphics[width=0.92\linewidth]{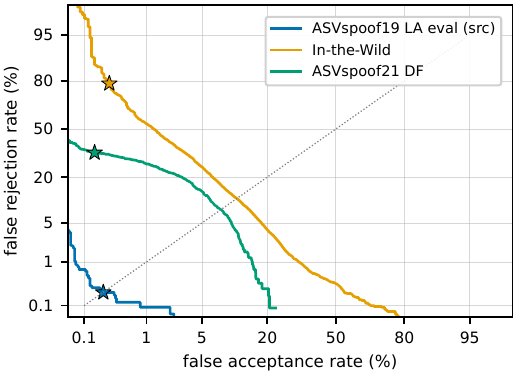}
\caption{DET curves of the frozen CM on the source set (ASVspoof~2019~LA eval) and both target sets (In-the-Wild, ASVspoof~2021~DF). Stars mark the operating point actually reached on each set when the LA-calibrated threshold $\tau_{\mathrm{LA}}$ is transferred; it lies far from the EER point (the diagonal crossing) of each target curve. This distance is the deployment failure that EER alone does not reveal.}
\label{fig:det}
\end{figure}

\begin{table}[t]
\caption{Cross-dataset replication on ASVspoof~2021~DF (30{,}592-utterance subset; 3.8\% bona fide prior). Protocol identical to Table~\ref{tab:itw}.}
\label{tab:df21}
\centering
\footnotesize
\setlength{\tabcolsep}{3.5pt}
\begin{tabular}{@{}lcccc@{}}
\toprule
Method & EER\,\% & HTER\,\% & FRR\,\% & FAR\,\% \\
\midrule
baseline (no correction) & 7.92 & 17.05 & 33.9 & 0.2 \\
C2 temperature/shift & 7.92 & 50.00 & 0.0 & 100.0 \\
C4 CORAL & 7.92 & \textbf{15.46} & 30.5 & 0.5 \\
C6 AS-norm (src bona, $k{=}100$) & 7.92 & 18.08 & 35.8 & 0.4 \\
\bottomrule
\end{tabular}
\end{table}

Table~\ref{tab:df21} and Fig.~\ref{fig:det} show that the audit phenomenon is not an ITW artifact: the in-domain 0.21\% becomes 7.92\% EER on DF21, with a transferred-threshold HTER of 17.1\% (33.9\% FRR). EER is again immobile across all corrections (identical to two decimals).

\textbf{Failure mode III: prior sensitivity.} The best-performing method on ITW, C2, collapses on DF21: its symmetric quartile pseudo-labels implicitly assume a roughly balanced target, but DF21 is 96\% spoof, so the fitted calibration degenerates ($T\to0$) and the system accepts everything (FRR 0\%, FAR 100\%, HTER 50\%); a fully converged refit of the same objective lands elsewhere in the degenerate region but still brings no improvement over the DF21 baseline (HTER 17.4\% vs.\ 17.1\%). An unlabeled correction that halves the deployment error on one target can thus be worse than no correction on another, purely because of the unknown class prior---a quantity that the label-shift literature requires to be estimated explicitly before any output adjustment \cite{saerens2002,lipton2018}. CORAL, by contrast, is the only audited method that improves HTER on DF21 (17.1\%$\to$15.5\%, 9.3\% relative)---while it degraded HTER on ITW. No audited correction is consistently safe across both targets.

\subsection{Why operating-point repair is fragile}
\label{ssec:discussion}

C2 deserves a closer analysis, because it produces both the largest deployment gain (ITW) and the worst collapse (DF21). Its pseudo-classes are, by construction, separated by the inter-quartile gap of the target score distribution, so its logistic likelihood has no finite minimizer: the loss decreases monotonically as $T\to0$ while the decision boundary $-b$ is free to settle anywhere inside the gap. The fitted calibration is therefore determined by the optimizer trajectory rather than by a population quantity, and what the audit measures is the behavior of the whole admissible family. On ITW the inter-quartile gap brackets the true class boundary, so every admissible endpoint is benign: our fixed L-BFGS run yields HTER 24.39\% and a fully converged refit 25.18\%. On DF21 the 96\% spoof prior places both quartiles inside the spoof mass; every admissible boundary is then degenerate, and the two endpoints we measured (HTER 50.0\% and 17.4\%, against a 17.1\% baseline) merely select different members of a family that contains no improvement at all. Prior shift does not just degrade C2---it removes the well-posedness of its fitting problem.

Viewed jointly, the seven corrections factor into three groups. The \emph{monotone} group (C1--C3) cannot move EER and only relocates the operating point; whether the relocation helps is decided by whether the matched statistic is class-balanced. Quartile pseudo-labels (C2) encode an implicit balanced prior, while global moments (C1, C3) align the class \emph{mixture} rather than the bona fide class: on ITW they move the operating point the wrong way, raising FRR from 78.7\% to 89.9\% and 93.2\%, respectively. The \emph{geometry} group (C4, C7) is non-monotone but acts far from the decision boundary: CORAL shifts EER by under 1\% relative yet is the only method improving HTER on DF21, while BN adaptation is a near no-op whose in-domain LA EER stays at 0.21\%---evidence that the LA$\to$ITW gap resides in the score geometry rather than in back-end feature statistics. The \emph{cohort} group (C5, C6) is catastrophic or inert depending solely on cohort composition. No group improves both targets, and no two groups fail in the same way---which is precisely why a single EER number cannot anticipate deployment behavior.

\subsection{Limitations}
\label{ssec:limit}

This audit covers one (state-of-the-art, widely used) CM; the invariance result (Prop.~\ref{prop:mono}) is model-agnostic, but the magnitudes of the non-monotone effects may vary across architectures. The LA eval and DF21 sets are fixed public subsets (4/9 and 4/80 mirror shards, respectively; labels of the LA subset verified against the official protocol), while ITW is used in full; absolute numbers on the full DF21 protocol may differ. C2 was audited with one pseudo-labeling rule (symmetric quartiles); prior-aware variants might be more robust, which we view as the constructive question our audit motivates rather than answers. HTER at the EER-transferred threshold is a single operating point: applications weighting FRR and FAR unequally will observe different absolute numbers, although Prop.~\ref{prop:mono} applies verbatim to any threshold-based cost. The discrete-crossing sensitivity quantified in Sec.~\ref{ssec:setup} (under one HTER point) bounds the protocol's own contribution to the reported gaps. Finally, DF21 contains codec and compression variability that we treat as part of the domain shift rather than as a separately controlled factor.

\section{Conclusions and Recommendations}
\label{sec:concl}

Auditing a frozen state-of-the-art CM yields three findings. First, EER hides deployment failure: a 0.21\%-EER detector operated at its source threshold rejects 78.7\% of genuine In-the-Wild speech while its target EER still reads 11.2\%. Second, the most common unlabeled fixes are structurally unable to repair EER (Prop.~\ref{prop:mono}), and the non-monotone ones that could, do not (under 1\% relative on ITW; none on DF21); genuine ranking repair appears to require adaptation beyond score and statistics matching. Third, operating-point repair is where unlabeled calibration genuinely helps, but it is fragile: its benefit inverts under class-prior shift, and ASV-style cohort normalization fails when the cohort mixes classes.

We recommend that CM evaluations (i)~report HTER (or the FRR/FAR pair) at a threshold transferred from the training-domain operating point, alongside EER; (ii)~treat any claimed EER gain obtained from a monotone score transform as a red flag; (iii)~stress-test unlabeled calibration under at least two target class priors; (iv)~when score normalization is used, report the cohort's composition and verify that its statistics are class-pure; and (v)~check that an unlabeled calibration objective remains well-posed under the deployment prior before trusting its output.

\bibliographystyle{IEEEbib}
\bibliography{refs}

\end{document}